\documentclass[12pt]{article}
\usepackage{fullpage}
\usepackage{url}
\usepackage{setspace}

\usepackage{graphicx}
\usepackage{subfigure}
\usepackage[nomarkers,nolists]{endfloat}

\newcommand{\shortcite}{\cite}
\usepackage{amsmath}
\usepackage{amsthm}
\newtheorem{proposition}{Proposition} 
\newtheorem{lemma}{Lemma}
\newtheorem{theorem}{Theorem}

\title{Information In The Non-Stationary Case}
\author{Vincent Q. Vu\footnotemark[2], 
                Bin Yu\footnotemark[2], 
                Robert E. Kass\footnotemark[3] \\
                {\footnotesize \tt \{vqv, binyu\}@stat.berkeley.edu, kass@stat.cmu.edu }\\
                \footnotemark[2]
                {\footnotesize Department of Statistics, University of California, Berkeley} 
                \\
                \footnotemark[3]
                {\footnotesize Department of Statistics and Center for the
                  Neural Basis of Cognition, Carnegie Mellon University}
}

\begin{document}
%%%%%%%%%%%%%%%%%%%%%%%%%%%%%%%%%%%%%%%%%%%%%%%%%%%%%%%%%%%%%%%%%%%%%%%%%%%%%%
%%%%%%%%%%%%%%%%%%%%%%%%%%%%%%%%%%%%%%%%%%%%%%%%%%%%%%%%%%%%%%%%%%%%%%%%%%%%%%
%%%%%%%%%%%%%%%%%%%%%%%%%%%%%%%%%%%%%%%%%%%%%%%%%%%%%%%%%%%%%%%%%%%%%%%%%%%%%%

\pagenumbering{alph}
\begin{titlepage}
\maketitle
\doublespacing
\begin{abstract}
%%% TEXEXPAND: INCLUDED FILE MARKER ./abstract.tex
Information estimates such as the ``direct method'' of Strong et al.
(1998) sidestep the difficult problem of estimating the joint 
distribution of response and stimulus by instead estimating 
the difference between the marginal and conditional entropies
of the response. While this is an effective estimation
strategy, it tempts the practitioner to ignore the
role of the stimulus and the meaning of mutual information.
We show here that, as the number of trials increases indefinitely,
the direct (or ``plug-in'') estimate of marginal entropy 
converges (with probability 1) to the entropy of the time-averaged 
conditional distribution of the response, and the direct estimate
of the conditional entropy converges to the time-averaged entropy 
of the conditional distribution of the response. Under joint stationarity
and ergodicity of the response and stimulus, the difference of
these quantities converges to the mutual information. When the
stimulus is deterministic or non-stationary the direct estimate
of information no longer estimates mutual information, which
is no longer meaningful, but it remains a measure of variability 
of the response distribution across time.
%%% TEXEXPAND: END FILE ./abstract.tex
\end{abstract}
\thispagestyle{empty}
\end{titlepage}

%%%%%%%%%%%%%%%%%%%%%%%%%%%%%%%%%%%%%%%%%%%%%%%%%%%%%%%%%%%%%%%%%%%%%%%%%%%%%%
%%%%%%%%%%%%%%%%%%%%%%%%%%%%%%%%%%%%%%%%%%%%%%%%%%%%%%%%%%%%%%%%%%%%%%%%%%%%%%
%%%%%%%%%%%%%%%%%%%%%%%%%%%%%%%%%%%%%%%%%%%%%%%%%%%%%%%%%%%%%%%%%%%%%%%%%%%%%%
\clearpage\pagenumbering{arabic}
\doublespacing

\section{Introduction}
\label{sec:introduction}
%%% TEXEXPAND: INCLUDED FILE MARKER ./intro.tex
Information estimates are used to characterize the amount of information
that a spike train contains about a stimulus \cite{Strong:1998,Borst:1999}.
They are motivated by information theory \cite{Shannon:1948} and widely believed to estimate the
mutual information (or mutual information rate) between stimulus and spike train response.
They are frequently calculated using data from experiments where the stimulus and response are  
dynamic and time-varying \shortcite{Hsu:2004,Reich:2001,Reinagel:2000,Nirenberg:2001}.

For mutual information to be properly defined, see for example \cite{Cover:1991aa}, 
the stimulus and response must be considered random, 
and when the estimates are obtained from time-averages, they should also be stationary and ergodic.
In practice these assumptions are usually tacit, and information estimates, 
such as the \emph{direct method} proposed by \cite{Strong:1998}, 
can be made without explicit consideration of the stimulus.
This can lead to misinterpretation.

The purpose of this note is to show that the direct method information estimate can be reinterpreted as 
the average divergence across time of the conditional response distribution from its overall mean; 
in the absence of stationarity and ergodicity: 
\begin{enumerate}
	\item information estimates do not necessarily estimate mutual information, but
	\item potentially useful interpretations can still be made by referring back to the time-varying divergence.
\end{enumerate}
Although our results are specialized to the direct method with the plug-in entropy estimator, 
they should hold more generally regardless of the choice of entropy estimator.
\footnote{See \cite{Victor2006Approaches-to-I} for a recent review of existing entropy estimators.}

The fundamental issue concerns stationarity: methods that assume stationarity are unlikely to be
appropriate when stationarity appears to be violated. In the non-stationary case, our second result should
be of use, as would be other methods that explicitly consider the dynamic and non-stationary nature of the
stimulus and response; see for instance \cite{Barbieri:2004}.

We begin with a brief review of the direct method and plug-in entropy estimator.
This is followed by results showing that the information estimate can be recast as a time-average.
This characterization leads us to the interpretation that the information estimate is actually a measure 
of variability of the stimulus conditioned response distribution.
This observation is first made in the finite number of trials case, and then formalized by 
a theorem describing the limiting behavior of the information estimate as the number of trials 
tends to infinity.  
Following the theorem is discussion about the interpretation of the limit, and examples that illustrate the interpretation with a proposed graphical plot.

\section{Review of the direct method}
In the direct method a time-varying stimulus is chosen by the experimenter and then
repeatedly presented to a subject over multiple trials. The observed responses are
conditioned by the same stimulus.
Two types of variation in the response are considered:
\begin{enumerate}
	\item variation across time (potentially related to the stimulus), and
	\item trial-to-trial variation.
\end{enumerate}
Figure~1(a) shows an example of data from such an experiment.
The upper panel is a raster plot of the response of a Field L neuron of an adult male Zebra Finch during synthetic song stimulation.  
The lower panel is a plot of the audio signal corresponding to the natural song. 
Details of the experiment can be found in \cite{Hsu:2004}.

Let us consider the random process $\{ S_t, R_t^k \}$ representing the value of the stimulus
and response at time $t=1,\ldots,n$ during trial $k=1,\ldots, m$.
The response is made discrete by dividing time into bins of size $dt$ and then considering \emph{words} (or patterns) of spike counts formed within intervals (overlapping or non-overlapping) of $L$ adjacent time bins.
The number of spikes that occur in each time bin become the letters in the words.
$R_t^k$ corresponds to these words, and may belong to a countably infinite set (because the number of spikes in a bin is theoretically unbounded).
In the raster plot of Figure~1(a) the time bin size is $dt = 1$ millisecond, and the vertical lines demarcate non-overlapping words of length $L=10$ time bins.

Given the responses $\{ R_t^k\}$, the direct method considers two different entropies:
\begin{enumerate}
	\item the \emph{total entropy} $H$ of the response, and
	\item the local \emph{noise entropy} $H_t$ of the response at time $t$.
\end{enumerate}
The total entropy is associated with the stimulus conditioned distribution of the response across all
times and trials. The local noise entropy is associated with the stimulus conditioned distribution of
the response at time $t$ across all trials.  
These quantities are calculated directly from the neural response, 
and the difference between the total entropy and the average
(over $t$) noise entropy is what \cite{Strong:1998} call ``the information that the spike train
provides about the stimulus.''

% VQV: Reviewer #1, Major comment #1
$H$ and $H_t$ depend implicitly on the length $L$ of the words.  Normalizing by $L$ and considering large 
$L$ leads to the total and local entropy rates 
that are defined to be $\lim_{L\to\infty} H(L) / L$ and $\lim_{L\to\infty} H_t(L) / L$, respectively, when they exist.  
The direct method of \shortcite{Strong:1998} prescribed an extrapolation for estimating these limits, 
however they do not necessarily exist when the stimulus and response process are non-stationary. 
When there is stationarity, estimation of entropy for large $L$ is potentially difficult, and extrapolation from a few small choices of $L$ can be suspect.  
Since we are primarily interested in the non-stationary case, we do not address these issues 
and refer the reader to \shortcite{Kennel:2005,Gao:2006} for larger discussion on the stationary case.
For notational simplicity, the dependence on $L$ will be suppressed in the remainder of the text.

\paragraph{The plug-in entropy estimate}
\cite{Strong:1998} proposed estimating $H$ and $H_t$ by plug-in with the
corresponding empirical distributions:
\begin{equation}
	\label{eq:phat}
	\hat{P}(r) 
	:= \frac{1}{mn} \sum_{t=1}^n \sum_{k=1}^m 1_{\{R_t^k = r\}}
\end{equation}
and
\begin{equation}
	\label{eq:pthat}
	\hat{P}_t(r)
	:= \frac{1}{m} \sum_{k=1}^m 1_{\{R_t^k = r\}}.
\end{equation}
% VQV: Reviewer #1, Minor comment #3
Note that $\hat{P}$ is also the average of $\hat{P}_t$ across $t=1,\ldots,n$.
So the direct method \emph{plug-in} estimates\footnote{\cite{Strong:1998} used the name \emph{naive
estimates}.} of $H$ and $H_t$ are
\begin{equation}
	\hat{H} 
	:= -\sum_r \hat{P}(r) \log \hat{P}(r),
\end{equation}
and 
\begin{equation}
	\hat{H}_t
	:= -\sum_r \hat{P}_t(r) \log \hat{P}_t(r),
\end{equation}
respectively.
The direct method plug-in information estimate is
\begin{equation}
	\label{eq:infodifference}
	\hat{I} := \hat{H} - \frac{1}{n} \sum_{t=1}^n \hat{H}_t .
\end{equation}
%%% TEXEXPAND: END FILE ./intro.tex

\section{Results}
\label{sec:results}
%%% TEXEXPAND: INCLUDED FILE MARKER ./results.tex
The direct method information estimate 
is not only the difference of entropies shown in (\ref{eq:infodifference}), 
but also a time-average of divergences.
The empirical distribution of response across all trials and times (\ref{eq:phat})
is equal to the average of $\hat{P}_t$ over time.
That is $\hat{P}(r) = n^{-1} \sum_{t=1}^n \hat{P}_t(r)$ and so
\begin{align}
	\hat{I} 
	&= \hat{H} - \frac{1}{n} \sum_{t=1}^n \hat{H}_t \\
	&= \frac{1}{n} \sum_{t=1}^n \sum_r \hat{P}_t(r) \log \hat{P}_t(r) - \sum_r \left[\frac{1}{n} \sum_{t=1}^n \hat{P}_t(r) \right] \log \hat{P}(r) \\
	&= \frac{1}{n} \sum_{t=1}^n \sum_r \hat{P}_t(r) \log \hat{P}_t(r) - \frac{1}{n} \sum_{t=1}^n \sum_r \hat{P}_t(r) \log \hat{P}(r) \\
	\label{eq:kldivergence}
	&= \frac{1}{n} \sum_{t=1}^n \sum_r \hat{P}_t(r) \log \frac{\hat{P}_t(r)}{\hat{P}(r)}.
\end{align}
The quantity that is averaged over time in (\ref{eq:kldivergence}) is the Kullback-Leibler divergence
between the empirical time $t$ response distribution $\hat{P}_t$ 
and the average empirical response distribution $\hat{P}$.

Since the same stimulus is repeatedly presented to the subject, and there is no evolution in the response, over multiple trials, 
the following \emph{repeated trial assumption} is natural:
\begin{quote}
	Conditional on the stimulus $\{S_t\}$ the $m$ 
	trials $\{S_t, R_t^1\}, \ldots, \{S_t, R_t^m\}$ are independent and identically 
	distributed (i.i.d.).
\end{quote}
Under this assumption $1_{\{R_t^1 = r\}}, \ldots, 1_{\{R_t^m = r\}}$ 
are conditionally i.i.d. for each fixed $t$ and $r$.  
Furthermore, the law of large numbers guarantees that as the number of 
trials $m$ increases the empirical response distribution $\hat{P}_t(r)$ converges to its conditional expected value
for each fixed $t$ and $r$.
Thus $\hat{P}_t(r)$ and $\hat{P}(r)$ can be viewed as estimates of $P_t(r |S_1,\ldots,S_n)$, 
defined by
\begin{equation}
	P_t(r |S_1,\ldots,S_n) := P(R_t^k = r | S_1,\ldots,S_n) = E\{\hat{P}_t(r) | S_1,\ldots,S_n\},
\end{equation}
and $\bar{P}(r | S_1,\ldots, S_n)$, defined by 
\begin{equation}
	\bar{P}(r | S_1,\ldots, S_n) 
	:= \frac{1}{n} \sum_{t=1}^n P_t(r | S_1,\ldots,S_n),
\end{equation}
respectively.  
$\bar{P}$ is average response distribution across time $t=1,\ldots,n$ conditional on the entire stimulus $\{S_1,\ldots,S_n\}$.

So the quantity that is averaged over time in (\ref{eq:kldivergence}) 
should be viewed as a plug-in estimate of the Kullback-Leibler divergence between $P_t$ and $\bar{P}$.  
We emphasize this by writing 
\begin{equation}
	\hat{D}(P_t || \bar{P})
	:= \sum_r \hat{P}_t(r) \log \frac{\hat{P}_t(r)}{\hat{P}(r)}.
\end{equation}
This observation will be formalized by the theorem of the next section.
For now we summarize the above with a proposition.
\begin{proposition}\label{pro:klidentity}
The information estimate is the time-average 
$\hat{I} = \frac{1}{n} \sum_{t=1}^n \hat{D}(P_t || \bar{P})$.
\end{proposition}
This decomposition of the information estimate is analogous to the decomposition 
of mutual information that \cite{Deweese:1999} call the ``specific surprise,''
while ``specific information'' is analogous to the alternative decomposition,
\begin{equation}
	\label{eq:specificinfo}
	\hat{I} = \frac{1}{n} \sum_{t=1}^n [ \hat{H} - \hat{H}_t ] .
\end{equation}
An important difference is that here the stimulus itself is a function of time and the 
decompositions are given in terms of time-dependent quantities.
It is possible that these quantities can reveal dynamic aspects of the stimulus and response 
relationship.  
This will be explored further in Sections~\ref{sub:interpretation_in_non_stationary_case} and \ref{subsec:visual}.

\subsection{What is being estimated?}
There are two directions in which the amount of observed response data can be increased:
length of time $n$, and number of trials $m$.  The information estimate is the average 
of $\hat{D}(P_t || \bar{P})$ over time, and may not necessarily converge as $n$ increases.
This could be due to $\{S_t, R_t^k\}$ being non-stationary and/or highly 
dependent in time.
Even when convergence may occur, the presence of serial correlation in $\hat{D}(P_t || \bar{P})$
(see the autocorrelation in panel (b) of Figures 2 for example)  
can make assessments of uncertainty in $\hat{I}$ difficult.

Assuming that the stimulus and response process is stationary and not too dependent in time  
could guarantee convergence, but this could be unrealistic.
On the other hand, the repeated trial assumption is appropriate if 
the same stimulus is repeatedly presented to the subject over multiple trials.
It is also enough to guarantee that the information estimate converges as the number of trials $m$ increases.
We prove the following theorem in the appendix.
\begin{theorem}\label{thm:information}
	Suppose that $P_t$ has finite entropy for all $t=1,\ldots, n$. 
	Then under the repeated trial assumption
	\begin{equation*}
		\lim_{m\to\infty} \hat{I} 
		= H(\bar{P}) - \frac{1}{n} \sum_{t=1}^n H(P_t) 
		= \frac{1}{n}\sum_{t=1}^n [H(\bar{P}) - H(P_t)]
		= \frac{1}{n}\sum_{t=1}^n D(P_t||\bar{P})
	\end{equation*} 
	with probability 1, 
	and in particular the following statements hold uniformly for $t=1,\ldots,n$ with probability 1:
	\begin{enumerate}
		\item $\lim_{m\to\infty} \hat{H} = H(\bar{P})$, 
		\item $\lim_{m\to\infty} \hat{H}_t = H(P_t)$, and 
		\item $\lim_{m\to\infty} \hat{D}(P_t || \bar{P}) = D(P_t || \bar{P})$ for $t=1,\ldots,n$, 
		\label{thm:information:kl}
	\end{enumerate}
	where $D(P_t || \bar{P})$ is the Kullback-Leibler divergence defined by, 
	\begin{equation*}
		D(P_t || \bar{P}) 
		:= \sum_r P_t(r |S_1,\ldots,S_n) \log \frac{P_t(r |S_1,\ldots,S_n)}{\bar{P}(r | S_1,\ldots, S_n)} ,
	\end{equation*}
	and $H(P)$ is the entropy of the distribution $P$, defined by 
	\begin{equation*}
		H(P) := -\sum_r P(r) \log P(r) .
	\end{equation*}
\end{theorem}
Note that if stationary and ergodicity do hold, then $P_t$ for $t=1,\ldots,n$ is also stationary 
and ergodic\footnote{$P_t$ and $\bar{P}$ are stimulus conditional distributions,
 and hence random variables potentially depending on $S_1,\ldots,S_n$.}.  
So its average, $\bar{P}(r)$, is guaranteed by the ergodic theorem to converge pointwise 
to $P(R_1^1 = r)$ as $n\to\infty$.  Moreover, if $R_1^1$ can only take on a finite number of values, then 
$H(\bar{P})$ also converges to the marginal entropy $H(R_1^1)$ of $R_1^1$.  
Likewise, the average of the conditional entropy $H(P_t)$ also converges to the 
expected conditional entropy: $\lim_{n\to\infty} H(R_n^1 | S_1, \ldots, S_n)$.  
So in this case the information estimate does indeed estimate mutual information.

However, the primary consequence of the theorem is that, in the absence of stationarity and ergodicity,
the information estimate $\hat{I}$ does not necessarily estimate mutual information.
The three particular statements show that the time-varying quantities $[\hat{H} - \hat{H}_t]$
and $\hat{D}(P_t || \bar{P})$ converge individually to the appropriate limits,
and justify our assertion that the information estimate is a time-average of 
plug-in estimates of the corresponding time-varying quantities.
Thus, the information estimate can always be viewed as an estimate of the time-average of 
either $D(P_t || \bar{P})$ or $[H(P) - H(P_t)]$--stationary and ergodic or not.

\subsection{The information estimate measures variability of the response distribution}
\label{sub:interpretation_in_non_stationary_case}

The Kullback-Leibler Divergence $D(P_t || \bar{P})$ has a simple interpretation:
it measures the dissimilarity of the time $t$ response distribution $P_t$ from 
its overall average $\bar{P}$.  So as a function of time, $D(P_t || \bar{P})$ 
measures how the conditional response distribution varies across time, relative to its overall mean.  
This can be seen in a more familiar form by considering the leading term of the Taylor expansion,
\begin{equation}
	D(P_t || \bar{P})
	=  
	\frac{1}{2} \sum_r 
	\frac{[P_t(r | S_1,\ldots, S_n) - \bar{P}(r | S_1,\ldots, S_n)]^2}{\bar{P}(r | S_1,\ldots, S_n)}
	+ \cdots.
\end{equation}
Thus, its average is in this sense a measure of the average variability of the 
response distribution.

It is, of course, possible that characteristics of the response are
due to confounding factors rather than the stimulus.  Furthermore, the
presence of additional noise in either process would weaken a measured
relationship between stimulus and response, compared to its strength
if the noise were eliminated. Setting these concerns aside, the
variation of the response distribution $P_t$ about its average
provides information about the relationship between the stimulus and
the response. In the stationary and ergodic case, this information may
be averaged across time to obtain mutual information. In more general
settings averaging across time may not provide a complete picture of
the relationship between stimulus and response. Instead, we suggest
examining the time-varying $D(P_t || \bar{P})$ directly, via graphical
display as discussed next.

\subsection{Plotting the divergence}
\label{subsec:visual}
The plug-in estimate $\hat{D}(P_t || \bar{P})$ is an obvious choice for estimating $D(P_t||\bar{P})$, 
but it turns out that estimating $D(P_t || \bar{P})$ is akin to estimating entropy.
Since the trials are conditionally i.i.d., the coverage adjustment method described in \cite{Vu:2007} 
can be used to improve estimation of $D(P_t || \bar{P})$ over the plug-in estimate.  
The appendix contains the details of this.

Figures 1 and 2 show the responses of the same Field L neuron of an adult male Zebra Finch 
under two different stimulus conditions.  
Details of the experiment and the statistics of the stimuli are described in \cite{Hsu:2004}.
Panel (a) of the figures shows the stimulus and response data.
In Figure~1 the stimulus is synthetic and stationary by construction, while in Figure~2 the stimulus is a natural song.
Panel (b) of the figures shows the coverage adjusted estimate of the divergence $D(P_t || \bar{P})$ plotted as a function of time.
95\% confidence intervals were formed by bootstrapping entire trials, i.e. an entire trial is 
either included in or excluded from a bootstrap sample.

The information estimate going along with each Divergence plot is the average of the solid curve
representing the estimate of $D(P_t || \bar{P})$.  
It is equal to 0.77 bits (per 10 millisecond word) in Figure~1(b) and 0.76 bits (per 10 millisecond word) in Figure~2(b).
Although the information estimates are nearly identical, the two plots are very different.

In the first case, the stimulus is stationary by construction and it appears that the 
time-varying divergence is too.  Its fluctuations appear to be roughly of the same scale across time, 
and its local mean is relatively stable.  The average of the solid curve seems to be a fair summary.

In the second case the stimulus is a natural song.  
The isolated bursts of the time-varying divergence and relatively flat regions in Figure~2(b) suggest 
that the response process (and the divergence) is non-stationary and has strong serial correlations.
The local mean of the divergence also varies strongly with time.
Summarizing $D(P_t ||\bar{P})$ by its time-average hides the time-dependent 
features of the plot.

More interestingly, when the divergence plot is compared to the plot of the stimulus in Figure~2,
there is a striking coincidence between the location of large isolated values of the 
estimated divergence and visual features of the stimulus waveform.  
They tend to coincide with the boundaries of the bursts in the stimulus signal.
This suggests that the spike train may carry information about the onset/offset of bursts 
in the stimulus.  
We discussed this with the Theunissen Lab and they confirmed from their STRF models 
that the cell in the example is an offset cell.
It tends to fire at the offsets of song syllables--the bursts of energy in the stimulus waveform.
They also suggested that a word length within the range of 30--50 milliseconds is a better match to the length of correlations in the auditory system.  
We regenerated the plots for words of length $L=40$ (not shown here) and found that the isolated structures in the divergence plot became even more pronounced.

%%% TODO: Is more explanation necessary?
%%% TEXEXPAND: END FILE ./results.tex

\section{Discussion}
\label{sec:discussion}
%%% TEXEXPAND: INCLUDED FILE MARKER ./discussion.tex
Estimates of mutual information, including the plug-in estimate, may
be viewed as measures of the strength of the relationship between the
response and the stimulus when the stimulus and response are jointly
stationary and ergodic. Many applications, however, use non-stationary
or even deterministic stimuli, so that mutual information is no longer
well defined. In such non-stationary cases do estimates of mutual
information become meaningless?  We think not, but the purpose of this
note has been to point out the delicacy of the situation, and to
suggest a viable interpretation of information estimates, along with
the divergence plot, in the non-stationary case.
 
In using stochastic processes to analyze data there is an implicit
practical acknowledgment that assumptions cannot be met precisely:
the mathematical formalism is, after all, an abstraction imposed on
the data; the hope is simply that the variability displayed by the
data is similar in relevant respects to that displayed by the
presumptive stochastic process. The ``relevant respects'' involve the
statistical properties deduced from the stochastic assumptions.  The
point we are trying to make is that highly non-stationary stimuli make
statistical properties based on an assumption of stationarity highly
suspect; strictly speaking, they become void.

To be more concrete, let us reconsider the snippet of natural song and
response displayed in Figure 2. When we look at the less than 2
seconds of stimulus amplitude given there, the stimulus is not at all
time-invariant: instead, the stimulus has a series of well-defined
bursts followed by periods of quiescence. Perhaps, on a very much
longer time scale, the stimulus would look stationary. But a good
stochastic model on a long time scale would likely require long-range
dependence.  Indeed, it can be difficult to distinguish
non-stationarity from long-range dependence \cite{Kunsch:1986}, and
the usual statistical properties of estimators are known to breakdown
when long-range dependence is present \cite{Beran:1994}.  Given a
short interval of data, valid statistical inference under stationarity
assumptions becomes highly problematic.  To avoid these problems we have
proposed the use of the divergence plot, and a recognition that the
``bits per second'' summary is no longer mutual information in the
usual sense. Instead we would say that the estimate of information
measures magnitude of variation of the response as the stimulus
varies, and that this is a useful assessment of the extent to which
the stimulus affects the response as long as other factors that affect
the response are themselves time-invariant.  In other deterministic or
non-stationary settings the argument for the relevance of an
information estimate should be analogous.  Under stationarity and
ergodicity, and indefinitely many trials, the stimulus sets that
affect the response---whatever they are---will be repeatedly sampled,
with appropriate probability, to determine the variability in the
response distribution, with time-invariance in the response being
guaranteed by the joint stationarity condition.  This becomes part of
the intuition behind mutual information.  In the deterministic or
non-stationary settings information estimates do not estimate mutual
information, but they may remain intuitive assessments of strength of
effect.
%%% TEXEXPAND: END FILE ./discussion.tex

\section*{Acknowledgments}
%%% TEXEXPAND: INCLUDED FILE MARKER ./acknowledgments.tex
The authors thank the Theunissen Lab at the University of California, Berkeley for providing the data set 
and helpful discussion.  
They also thank an anonymous reviewer for comments that greatly improved the manuscript.
V.~Q.~Vu was supported by a NSF VIGRE Graduate Fellowship and NIDCD grant DC 007293. B.~Yu
was supported by NSF grants DMS-03036508, DMS-0605165, DMS-0426227, ARO grant
W911NF-05-1-0104, NSFC grant 60628102, and a fellowship from the John Simon Guggenheim Memorial 
Foundation. This work began while Kass was a Miller Institute Visiting Research Professor at the 
University of California, Berkeley. Support from the Miller Institute is greatly appreciated. Kass's
work was also supported in part by NIMH grant RO1-MH064537-04.%%% TEXEXPAND: END FILE ./acknowledgments.tex

%%%%%%%%%%%%%%%%%%%%%%%%%%%%%%%%%%%%%%%%%%%%%%%%%%%%%%%%%%%%%%%%%%%%%%%%%%%%%
%%%%%%%%%%%%%%%%%%%%%%%%%%%%%%%%%%%%%%%%%%%%%%%%%%%%%%%%%%%%%%%%%%%%%%%%%%%%%
%%%%%%%%%%%%%%%%%%%%%%%%%%%%%%%%%%%%%%%%%%%%%%%%%%%%%%%%%%%%%%%%%%%%%%%%%%%%%
\appendix

\section{Appendix}
%%% TEXEXPAND: INCLUDED FILE MARKER ./appendix.tex
\subsection{Coverage adjusted estimate of $D(P_t || \bar P)$}
The main idea behind coverage adjustment is to adjust estimates for potentially 
unobserved values.  This happens in two places: estimation of $P_t$ and estimation of $D(P_t||\bar P)$.
In the first case, unobserved values affect the amount of weight that $\hat{P}_t$,
% VQV: Reviewer #1, Minor comment #6
defined in (\ref{eq:pthat}) in the main text,
 places on observed 
values.  In the second case unobserved values correspond to missing summands when plugging 
$\hat{P}_t$ into the Kullback-Leibler divergence.  
\cite{Vu:2007} gives a more thorough explanation of these ideas.  Let
\begin{equation}
	N_t(r) := \sum_{k=1}^m 1_{\{R_t^k = r\}}.
\end{equation}
The sample coverage, or total $P_t$-probability of observed values $r$, is estimated by 
$\hat{C}_t$ defined by 
\begin{equation}
	\hat{C}_t := 1 - \frac{\#\{r : N_t(r)=1 \} + .5}{m+1}.
\end{equation}
The number in the numerator of the fraction refers to the number of singletons---patterns that were observed only once across the $m$ trials at time $t$.
Then the coverage adjusted estimate of $P_t$ is the following shrunken version of $\hat{P}_t$:
\begin{equation}
	\tilde{P}_t(r) = \hat{C}_t \hat{P}_t(r).
\end{equation}
$\bar{P}$ is estimated by simply averaging $\tilde{P}_t$:
\begin{equation}
	\tilde{P}(r) = \frac{1}{n} \sum_{t=1}^n \tilde{P}_t(r).
\end{equation}
The coverage adjusted estimate of $D(P_t||\bar{P})$ is obtained by plugging  
$\tilde{P}_t$ and $\tilde{P}$ into the Kullback-Leibler divergence, 
but with an additional weighting on the summands according to 
the inverse of the estimated probability that the summand is observed:
\begin{equation}
	\tilde{D}(P_t || \bar{P})
	:= \sum_r \frac{\tilde{P}_t(r)\{\log\tilde{P}_t(r) - \log \tilde{P}(r)\}}{1-(1-\tilde{P}_t(r))^m}
	.
\end{equation}
The additional weighting is to correct for potentially missing summands.  (This is also explained in detail in \cite{Vu:2007}.)
Confidence intervals for $D(P_t || \bar{P})$ can be obtained by bootstrap sampling entire trials,
and applying $\tilde{D}$ to the bootstrap replicate data.

\subsection{Proofs}
We will use the following extension of the Lebesgue Dominated Convergence Theorem in the proof of 
Theorem~\ref{thm:information}.
\begin{lemma}
	\label{lem:dct}
	Let $f_m$ and $g_m$ for $m=1, 2, \ldots$ 
	be sequences of measurable, integrable functions defined on a measure space equipped with 
	measure $\mu$, and with pointwise limits $f$ and $g$, respectively. 
	Suppose further that $|f_m| \leq g_m$ and $\lim_{m\to\infty} \int g_m \ d\mu = \int g \ d\mu < \infty$.
	Then $$\lim_{m\to\infty} \int f_m \ d\mu = \int \lim_{m\to\infty} f_m  \ d\mu .$$
	\begin{proof}
		By linearity of the integral, 
		\begin{equation*}
			\liminf_{n\to\infty} \int (g + g_m) \ d\mu - \limsup_{n\to\infty} \int |f - f_m| \ d\mu
			= \liminf_{n\to\infty} \int (g + g_m) - |f - f_m| \ d\mu.
		\end{equation*}
		Since $0 \leq (g + g_m) - |f - f_m|$, Fatou's Lemma implies
		\begin{equation*}
			\liminf_{n\to\infty} \int (g + g_m) - |f - f_m| \ d\mu
			\geq \int \liminf_{n\to\infty} (g + g_m) - |f - f_m| \ d\mu .
		\end{equation*}
		The limit inferior on the inside of the right-hand integral is equal to $2g$ by assumption.
		Combining with the previous two displays and the assumption that 
		$\int g_m \ d\mu \to \int g \ d\mu$ gives 
		\begin{equation*}
			\limsup_{n\to\infty} |\int f  d\mu - \int f_m d\mu| 
			\leq \limsup_{n\to\infty} \int |f - f_m| d\mu 
			\leq 0.
		\end{equation*}
	\end{proof}
\end{lemma}

\begin{proof}[Proof of Theorem~\ref{thm:information}]
\label{pf:nonstationary}
%%% TEXEXPAND: INCLUDED FILE MARKER ./proof_convergence.tex
The main statement of the theorem is implied by the three numbered statements together with 
Proposition~\ref{pro:klidentity}.  We start with the second numbered statement.
Under the repeated trial assumption, $R_t^1,\ldots, R_t^m$ are conditionally i.i.d. 
given the stimulus $\{S_t\}$.  
So Corollary 1 of \cite{Antos2001Convergence-Pro}, can be applied to show that
\begin{align}
	\lim_{m\to\infty} \hat{H}_t
	\label{eq:sumconverges}
	&= \lim_{m\to\infty} -\sum_r \hat{P}_t(r) \log \hat{P}_t(r)	\\
	&= -\sum_r P_t(r|S_1,\ldots,S_n) \log P_t(r|S_1,\ldots,S_n) \\
	&= H(P_t)
\end{align}
with probability 1.  This proves the first numbered statement.

We will use Lemma~\ref{lem:dct} to prove the first numbered statement.
For each $r$ the law of large numbers asserts 
$\lim_{m\to\infty} \hat{P}_t(r) = P_t(r|S_1,\ldots,S_n)$ with probability 1.
So for each $r$, 
\begin{equation}
	\lim_{m\to\infty} -\hat{P}_t(r) \log \hat{P}(r)
	= -P_t(r|S_1,\ldots,S_n) \log \bar{P}(r|S_1,\ldots,S_n) 
\end{equation}
and 
\begin{equation}
	\lim_{m\to\infty} -\hat{P}_t(r) \log \hat{P}_t(r)
	= -P_t(r|S_1,\ldots,S_n) \log P_t(r|S_1,\ldots,S_n) 
	\label{eq:fm}
\end{equation}
with probability 1.  
Fix a realization where (\ref{eq:sumconverges}--\ref{eq:fm}) hold and let 
\begin{equation*}
	f_m(r) := -\hat{P}_t(r) \log \hat{P}(r)
\end{equation*}
and
\begin{equation*}
	g_m(r) := -\hat{P}_t(r)[\log \hat{P}_t(r) - \log n] .
\end{equation*}
Then for each $r$ 
\begin{equation*}
	\lim_{m\to\infty} f_m(r) = -P_t(r|S_1,\ldots,S_n) \log \bar{P}(r|S_1,\ldots,S_n) 
	=: f(r)
\end{equation*}
and
\begin{equation*}
	\lim_{m\to\infty} g_m(r) = - P_t(r)[\log P_t(r) - \log n]
	=: g(r).
\end{equation*}
The sequence $f_m$ is dominated by $g_m$ because
\begin{align}
	0 \leq -\hat{P}_t(r) \log \hat{P}(r) 
	&= f_m(r) \\
	&= -\hat{P}_t(r) [\log \sum_{u=1}^n \hat{P}_u(r) - \log n] \\
	\label{eq:logineq}
	&\leq -\hat{P}_t(r) [\log \hat{P_t}(r) - \log n] \\
	&= g_m(r)
\end{align}
for all $r$, where (\ref{eq:logineq}) uses the fact that $\log x$ is an increasing function.  
From (\ref{eq:sumconverges}) we also have that $\lim_{m\to\infty} \sum_r g_m(r) = \sum_r g(r)$.
Clearly, $f_m$ and $g_m$ are summable.  Moreover $H(P_t) < \infty$ by assumption.  So
\begin{equation}
	\sum_r g(r) 
	= \sum_r - P_t(r) \log P_t(r) + \log n \sum_r P_t(r) 
	= H(P_t) + \log n < \infty
\end{equation}
and the conditions of Lemma~\ref{lem:dct} are satisfied.
Thus
\begin{equation}
	\label{eq:kldenom}
	\lim_{m\to\infty} \sum_r -\hat{P}_t(r) \log \hat{P}(r)
	= \lim_{m\to\infty} \sum_r f_m(r)
	= \sum_r f(r)
	= \sum_r -P_t(r) \log \bar{P}(r).
\end{equation}
Averaging over $t=1,\ldots n$ gives
\begin{equation}
	\hat{H}
	= \lim_{m\to\infty} \sum_r -\hat{P}(r) \log \hat{P}(r)
	= \sum_r -\bar{P}(r) \log \bar{P}(r)
	= H(\bar P).
\end{equation}
for realizations where (\ref{eq:sumconverges}--\ref{eq:fm}) hold.  
This proves the first numbered statement because the probability of all such realizations is 1.

For the third numbered statement we begin with the expansions
\begin{equation}
	\hat{D}(P_t||\bar{P})
	= \sum_r \hat{P}_t(r) \log \hat{P}_t(r) - \hat{P}_t(r) \log \hat{P}(r).
\end{equation}
and 
\begin{equation}
	D(P_t||\bar{P})
	= \sum_r P_t(r) \log P_t(r) - P_t(r) \log \bar{P}(r).
\end{equation}
The second numbered statement and (\ref{eq:kldenom}) imply
\begin{equation}
	\lim_{m\to\infty} \sum_r \hat{P}_t(r) \log \hat{P}_t(r) - \hat{P}_t(r) \log \hat{P}(r)
	= \sum_r P_t(r) \log P_t(r) - \sum_r P_t(r) \log \bar{P}(r)
\end{equation}
with probability 1.  This proves the third numbered statement.%%% TEXEXPAND: END FILE ./proof_convergence.tex
\end{proof}
%%% TEXEXPAND: END FILE ./appendix.tex

% ============
% = Figure 1 =
% ============
\begin{figure}[p]
	\centering
	\subfigure[Stimulus and response]{\label{fig:syntheticexperiment}
	\includegraphics[width=.95\linewidth]{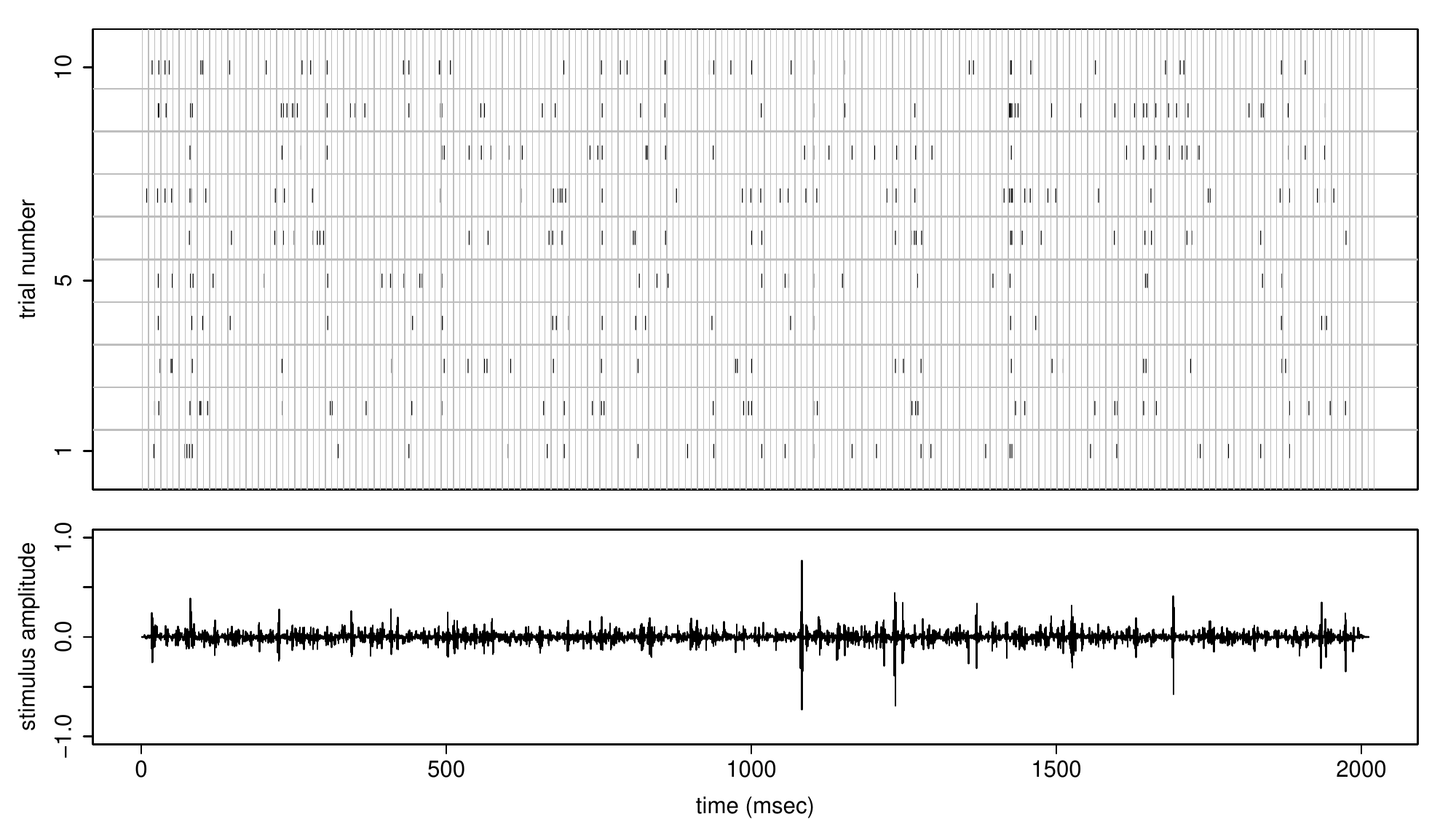}
	}
	\subfigure[Divergence plot]{\label{fig:syntheticdivergence}
	\includegraphics[width=.95\linewidth]{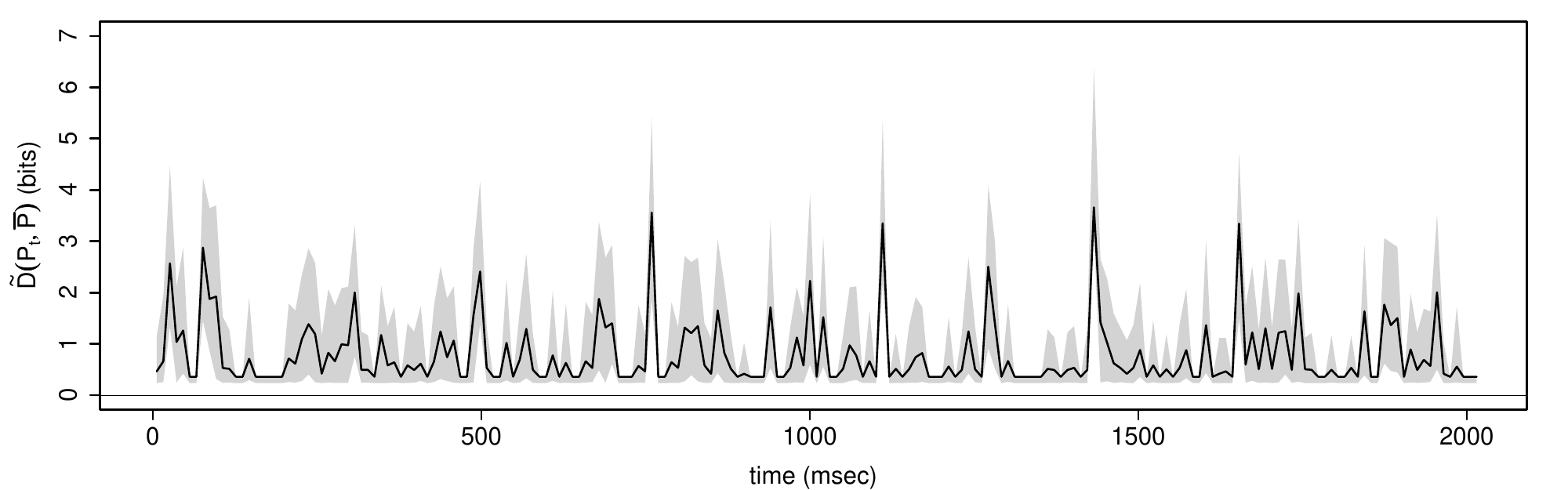}
	}

	\caption{
	\doublespacing
	(a) Raster plot of the response of the a Field L neuron of an adult male Zebra Finch (above) during the presentation
	of a synthetic audio stimulus (below) for 10 repeated trials.
	The vertical lines indicate boundaries of $L=10$ millisecond (msec) words formed at a 
	resolution of $dt = 1$ msec.  The data consists of 10 trials, each of duration 2000 msecs.
	(b) The coverage adjusted estimate (solid line) of $D(P_t,\bar{P})$ from the response shown above with 10 msec words.
	Pointwise 95\% confidence intervals are indicated by the shaded region and obtained by bootstrapping the trials 1000 times.
	The information estimate, 0.77 bits (per 10msec word, or 0.077 bits/msec), corresponds to the average value of the solid curve.
	}
\end{figure}

% ============
% = Figure 2 =
% ============
\begin{figure}[p]
	\centering
	\subfigure[Stimulus and response]{\label{fig:naturalexperiment}
	\includegraphics[width=.95\linewidth]{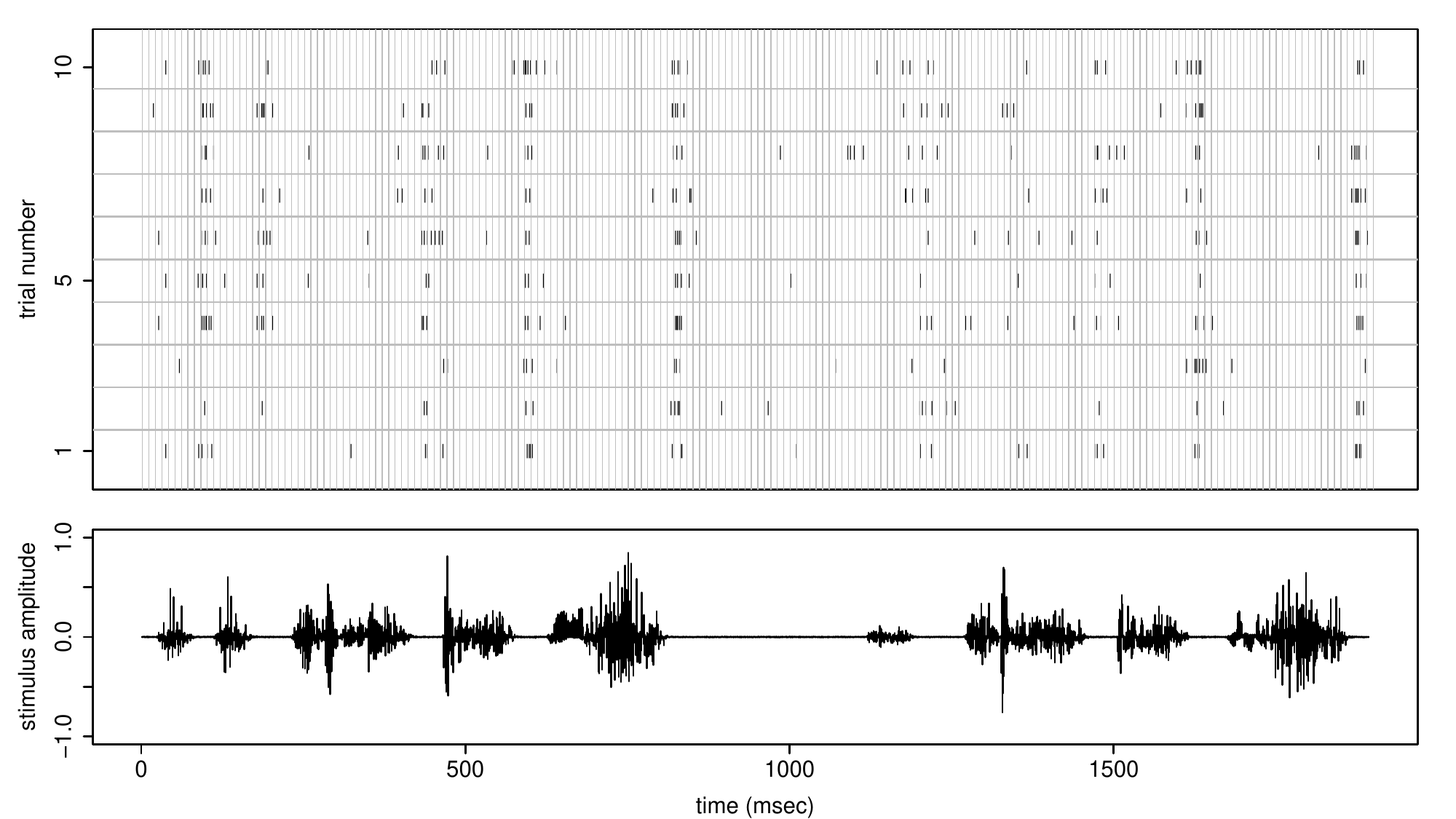}
	}
	\subfigure[Divergence plot]{\label{fig:naturaldivergence}
	\includegraphics[width=.95\linewidth]{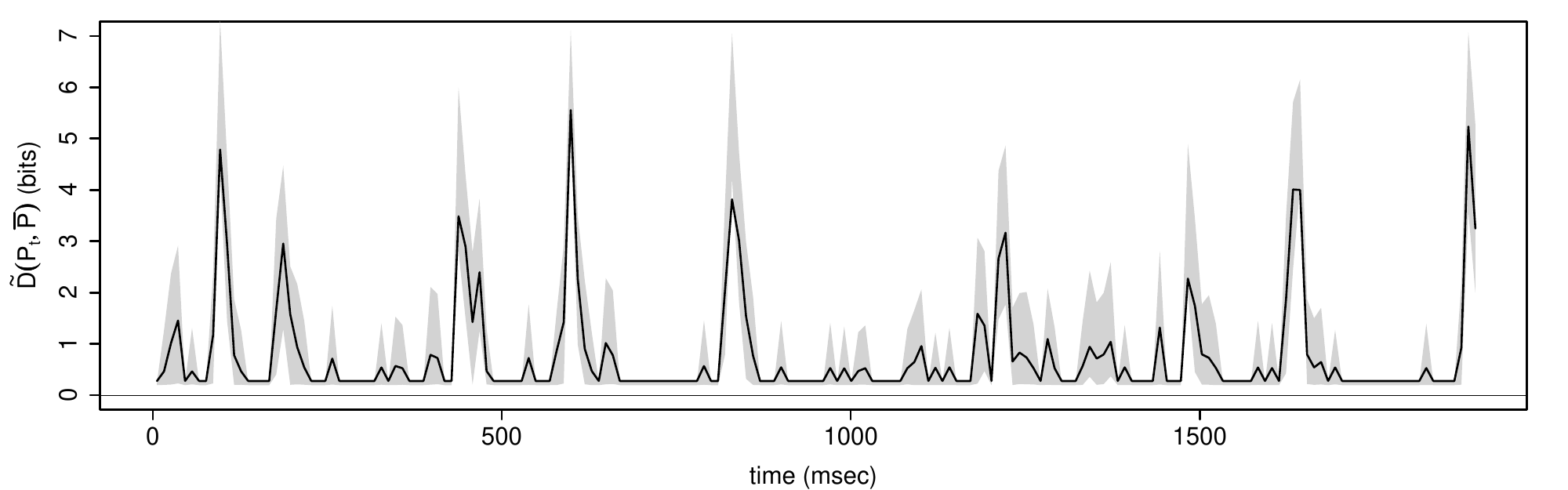}
	}

	\caption{
	\doublespacing
	(a) Same as in Figure~1, but in this set of trials the stimulus
	is a conspecific natural song.
	(b) The coverage adjusted estimate (solid line) of $D(P_t,\bar{P})$ from the response shown above.
	Pointwise 95\% confidence intervals are indicated by the shaded region and obtained by bootstrapping the trials 1000 times.  
	The information estimate, 0.76 bits (per 10 msec word or 0.076 bits/msec), corresponds to the average value of the solid curve.
	}
\end{figure}


\begin{thebibliography}{10}

\bibitem{Antos2001Convergence-Pro}
Andr{\'a}s Antos and Ioannis Kontoyiannis.
\newblock Convergence properties of functional estimates for discrete
  distributions.
\newblock {\em Random Structures and Algorithms}, 19:163--193, 2001.

\bibitem{Barbieri:2004}
Riccardo Barbieri, Loren~M. Frank, David~P. Nguyen, Michael~C. Quirk, Victor
  Solo, Matthew~A. Wilson, and Emery~N. Brown.
\newblock Dynamic analyses of information encoding in neural ensembles.
\newblock {\em Neural Computation}, 16(2):277--307, 2004.

\bibitem{Beran:1994}
Jan Beran.
\newblock {\em Statistics for long-memory processes}.
\newblock Chapman \& Hall Ltd., 1994.

\bibitem{Borst:1999}
Alexander Borst and Fr{\'e}d{\'e}ric~E. Theunissen.
\newblock Information theory and neural coding.
\newblock {\em Nature Neuroscience}, 2(11):947--957, 1999.

\bibitem{Cover:1991aa}
T.~Cover and J.~Thomas.
\newblock {\em Elements of Information Theory}.
\newblock Wiley, New York, 1991.

\bibitem{Deweese:1999}
M~Deweese and M~Meister.
\newblock How to measure the information gained from one symbol.
\newblock {\em Network: Computation in Neural Systems}, Jan 1999.

\bibitem{Gao:2006}
Yun Gao, Ioannis Kontoyiannis, and Elie Bienenstock.
\newblock From the entropy to the statistical structure of spike trains.
\newblock {\em Information Theory, 2006 IEEE International Symposium on}, pages
  645--649, July 2006.

\bibitem{Hsu:2004}
Anne Hsu, Sarah M~N Woolley, Thane~E Fremouw, and Fr{\'e}d{\'e}ric~E
  Theunissen.
\newblock Modulation power and phase spectrum of natural sounds enhance neural
  encoding performed by single auditory neurons.
\newblock {\em J. Neuro.}, 24(41):9201--9211, 2004.

\bibitem{Kennel:2005}
Matthew~B Kennel, Jonathon Shlens, Henry D~I Abarbanel, and E~J Chichilnisky.
\newblock Estimating entropy rates with bayesian confidence intervals.
\newblock {\em Neural Computation}, 17(7):1531--1576, 2005.

\bibitem{Kunsch:1986}
H~Kunsch.
\newblock Discrimination between monotonic trends and long-range dependence.
\newblock {\em Journal of Applied Probability}, 23(4):1025--1030, Jan 1986.

\bibitem{Nirenberg:2001}
S~Nirenberg, S~M Carcieri, A~L Jacobs, and P~E Latham.
\newblock Retinal ganglion cells act largely as independent encoders.
\newblock {\em Nature}, 411(6838):698--701, Jun 2001.

\bibitem{Reich:2001}
Daniel~S. Reich, Ferenc Mechler, and Jonathan~D. Victor.
\newblock Formal and attribute-specific information in primary visual cortex.
\newblock {\em Journal of Neurophysiology}, 85(1):305--318, 2001.

\bibitem{Reinagel:2000}
Pamela Reinagel and R.~Clay Reid.
\newblock Temporal coding of visual information in the thalamus.
\newblock {\em Journal of Neuroscience}, 20(14):5392--5400, 2000.

\bibitem{Shannon:1948}
C.~E. Shannon.
\newblock A mathematical theory of communication.
\newblock {\em Bell System Technical Journal}, 27:379--423, 1948.

\bibitem{Strong:1998}
S.~P. Strong, Roland Koberle, Rob de~Ruyter~van Steveninck, and William Bialek.
\newblock Entropy and information in neural spike trains.
\newblock {\em Physical Review Letters}, 80(1):197--200, 1998.

\bibitem{Victor2006Approaches-to-I}
Jonathon~D. Victor.
\newblock Approaches to information-theoretic analysis of neural activity.
\newblock {\em Biological Theory}, 1:302--316, 2006.

\bibitem{Vu:2007}
Vincent~Q. Vu, Bin Yu, and Robert~E. Kass.
\newblock Coverage adjusted entropy estimation.
\newblock {\em Statistics in Medicine}, 26(21):4039--4060, 2007.

\end{thebibliography}
\end{document}